\documentclass[10pt,twocolumn,twoside]{IEEEtran} 

\usepackage{cite}  
\usepackage{times}
\usepackage{amsmath,amsfonts,amssymb,mathtools}

\usepackage{epsfig,amstext,xspace}

\usepackage{amsthm}

\usepackage{graphicx} \usepackage{caption} \usepackage{subcaption}

\usepackage{algpseudocode}
\usepackage{multirow}
\usepackage[table]{xcolor}

\usepackage{algorithm}
\usepackage{array}
\newcommand{\raf}[1]{(\ref{#1})}
\newtheorem{theorem}{Theorem}

\newtheorem{lemma}[theorem]{Lemma}

\newtheorem{definition}{Definition}

\newtheorem{remark}{Remark}

\newcommand{\OPT}{\ensuremath{\textsc{Opt}}}

\def\CC{\mathbb C}
\def\RR{\mathbb R}

\def\cA{\mathcal A}

\def\cF{\mathcal F}

\def\cI{\mathcal I}

\def\cN{\mathcal N}
\def\cS{\mathcal S}
\def\cT{\mathcal T}

\def\cU{\mathcal U}

\def\bzero{\mathbf 0}

\newcommand{\hide}[1]{}

\newcommand{\re}{\ensuremath{\mathrm{Re}}}
\newcommand{\im}{\ensuremath{\mathrm{Im}}}

\newcommand{\cP}{\ensuremath{\mathcal{P}}}

\newcommand{\cV}{\ensuremath{\mathcal{V}}}
\newcommand{\cE}{\ensuremath{\mathcal{E}} }

\title{Combinatorial Optimization of AC Optimal Power Flow with Discrete Demands in Radial Networks}

\author{
Majid Khonji, Sid Chi-Kin Chau, and Khaled Elbassioni
\thanks{M. Khonji and K. Elbassioni are with the Department of Electrical Engineering and Computer Science at Khalifa University of Science and Technology.  S. C.-K. Chau is with the Research School of Computer Science, Australian National University.  (e-mail: \{majid.khonji, khaled.elbassioni\}@ku.ac.ae, sid.chau@anu.edu.au).}\vspace*{-20pt}
}

\newif\ifsupplementary
\supplementaryfalse

\begin{document}

\maketitle

\begin{abstract}
The AC Optimal power flow (OPF) problem is one of the most fundamental problems in power systems engineering. For the past decades, researchers have been relying on unproven heuristics to tackle OPF. The hardness of OPF stems from two issues: (1) non-convexity  and (2) combinatoric constraints (e.g., discrete power extraction constraints). The recent advances in providing sufficient conditions on the exactness of convex relaxation of OPF can address the issue of non-convexity. To complete the understanding of OPF, this paper presents a polynomial-time approximation algorithm to solve the convex-relaxed OPF with discrete demands as combinatoric constraints, which has a provably small parameterized approximation ratio (also known as PTAS algorithm).  Together with the sufficient conditions on the exactness of the convex relaxation, we provide an efficient approximation algorithm to solve OPF with discrete demands,  when the underlying network is radial with a fixed size and one feeder. 
{\color{black}The running time of PTAS is $O(n^{4m/\epsilon}T)$, where $T$ is the time required to solve a convex relaxation of the problem, and $m, \epsilon$ are fixed constants. Based on prior hardness results of OPF, our PTAS is among the best achievable in theory.}
Simulations show our algorithm can produce close-to-optimal solutions in practice.

\end{abstract}
\begin{IEEEkeywords}
Optimal power flow, approximation algorithms, discrete power demands, combinatorial optimization, PTAS
\end{IEEEkeywords}

\pagenumbering{arabic}

\vspace{-10pt}
\section{Introduction}

The AC optimal power flow (OPF) problem underpins many optimization problems of power systems. However, OPF is notoriously hard to solve. The hardness of OPF stems from two issues: (1) {\em non-convexity} because of the non-convex constraints involving complex-valued entities of  power systems,  and (2) {\em combinatoric constraints}, for example, discrete power injection/extraction constraints. In the past, due to the lack of understanding of solvability of OPF, researchers have been relying on unproven heuristics or general numerical solvers, which suffer from the issues of excessive running-time, lack of termination guarantee, or uncertainty of how far the output solutions deviate from the true optimal solutions.

Recently, there have been advances in tackling OPF by applying convex relaxations\cite{huang2017sufficient,gan2015exact,low2014convex1,low2014convex2}. These results imply that the relaxation of certain equality operating constraints to be inequality constraints can attain a more tractable convex programming problem which admits an optimal solution to the original problem, under certain mild sufficient conditions verifiable in a prior. Remarkably, these results can be applied to OPF with discrete power extractions (e.g., \cite{gan2015exact}). 
 There are also several relaxations based on semi-definite programming (SDP) formulations that rely on the bus injection model for power flow analysis. \cite{kocuk2016strong, hijazi2017convex, fattahi2017conic, josz2015application}.  

However, because of the lack of proper efficient algorithms to solve the convex-relaxed OPF with combinatoric constraints, most prior papers (e.g., \cite{huang2017sufficient,gan2015exact,low2014convex1,low2014convex2}) only studied OPF with continuous power extraction constraints, such that the controls of power extractions can be partially satisfied. Some papers considered OPF  with discrete control variables \cite{briglia2017distribution,murray2015optimal,lin2008distributed, hijazi2017convex, fattahi2017conic}, but their  algorithms rely mainly on heuristics and show no guarantee on optimality and running time.

 In practice, there are many discrete power extraction constraints. For example, certain loads and devices can be either switched on or off, and hence, their control decision variables are binary.  To tackle OPF in these settings, it is important to provide feasible solutions that satisfy the discrete power extraction constraints. 

Solving combinatorial optimization problems by efficient algorithms in general is a main subject studied in theoretical computer science. Hence, we will draw on the related notions and terminology from theoretical computer science. There is a well-known class of problems, known as NP-hard problems, which are believed to be intractable to find the exact optimal solutions in polynomial-time. However, taking into consideration of approximation solutions (i.e., the solutions are within a certain approximation ratio over an optimal solution), it is possible to obtain efficient polynomial-time approximation algorithms for certain NP-hard problems. One efficient type of approximation algorithms is called PTAS ({\em polynomial-time approximation scheme}) \cite{Vazirani10}, which allows a parametrized approximation ratio as the running time of the algorithm. Thus, one can change the desired approximation ratio, at the expense of running time. In this paper,  our goal is to provide a PTAS to solve OPF with discrete demands as combinatoric constraints.  Due to the difficulty of the problem in its most general form, we will focus on radial networks, and present our PTAS for the case when the underlying network has a {\it fixed} size and only one generator. {\color{black} The running time of PTAS is $O(n^{4m/\epsilon}T)$, where $T$ is the time required to solve a convex relaxation of the problem, $m$ is the network size, and $\epsilon$ is a fixed constant that controls the approximation ratio.}


This work is also related to a number of recent developments.  First, the combinatorial optimization for a single-capacity power system has been studied  as {\em complex-demand knapsack problem} in prior work \cite{chau2016truthful}. Then, an approximation algorithm is provided for simplified DistFlow model of OPF without considering generation cost in \cite{khonji2016optimal}. To our best knowledge, this is the first work to present a PTAS algorithm for combinatorial optimization in a realistic OPF model.

This paper is structured as follows: First, the model of OPF and the idea of convex relaxation are reviewed in Secs.~\ref{sec:model}-\ref{sec:relax}. Next, we presents a PTAS algorithm to solve the convex-relaxed OPF with discrete demands in Sec.~\ref{sec:ptas}. Together with the sufficient conditions on the exactness of convex relaxation, we provide an efficient approximation algorithm to solve OPF with discrete demands. Furthermore, fundamental hardness results of OPF are discussed to show that our PTAS is among the best achievable in Sec.~\ref{sec:hard}. Lastly, simulations in Sec.~\ref{sec:sim} show that the proposed algorithm can produce close-to-optimal solutions in practice.

\vspace*{-10pt}
\section{Problem Definition and Notations} \label{sec:model}

\subsection{Optimal Power Flow Problem on Radial Networks}

As in the previous work \cite{huang2017sufficient,gan2015exact}, this paper considers a radial (tree) electric distribution network, represented by a graph $\cT=(\cV,\cE)$. The set of nodes $\cV=\{0,...,m\}$ denotes the electric buses, and the set of edges $\cE$ denotes the distribution lines. Let $\cV^+\triangleq \cV \setminus\{0\}$. A substation feeder is attached to the root of the tree, denoted by node $0$. We assume that root $0$ is only connected to node 1 via a single edge (0,1). See an illustration in Fig.~\ref{fig:tree}. Since $\cT$ is a tree, $|\cV^+|=|\cE| = m$. Let $\cT_i=(\cV_i,\cE_i)$ be the subtree rooted at node $i$.
\begin{figure}[!htb]
	\begin{center} \vspace{-15pt}
		\includegraphics[scale=.6]{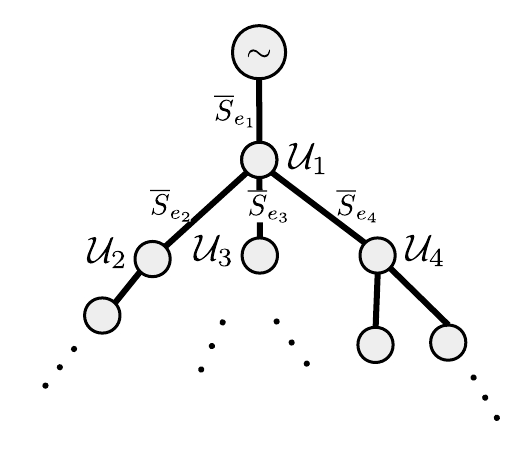}
	\end{center}\vspace{-20pt}
	\caption{An illustration of a radial network.}\vspace{-10pt}
	\label{fig:tree} 
\end{figure}

Note that this paper adopts the flow orientation that power flows from the root (node $0$) towards the leaves\footnote{Our formulation is different from \cite{gan2015exact, huang2017sufficient}, which assume the flow orientation that power flows from a leaf to the root. But it is shown in \cite{low2014convex1} that there is a bijection between the models of two flow orientations.}.  Hence, tuple $(i,j)$ refers to a directed edge, where node $i$ is a parent of $j$. Denote the path from the root $0$ to node $j$ by $\cP_j$.  

 Instead of assigning a single power extraction to each node, this paper considers a general setting where a set of users are attached to each node. Each user can control his power demand individually. Let $\cN$ be the set of all users, where $|\cN|=n$. Denote the set of  users  attached to node $j$ by $\cU_j \subseteq \cN$. Let the set of users within subtree $\cT_j$ be $\cN_j \triangleq \cup_{j\in\cV_j} \cU_j$. Denote the path from root $0$ to user $k$ by $\cP_k$.

The demand for user $k$ is represented by a complex number  $s_k\in \CC$.  
Among the users, some have discrete (inelastic) power demands, denoted by $\cI \subseteq\cN$. A discrete demand $s_k$, for $k\in \cI$, takes values from a discrete set $\hat \cS_k \subseteq \CC$. We assume $\hat \cS_k\triangleq \{\bzero,\overline s_k\}$, where a demand $s_k$ can be either completely satisfied at level $\overline s_k \in \CC$ or dropped, e.g., a piece of equipment that is either switched on with a fixed power rate or off\footnote{We make such assumption for simplicity. However, an extension to any set of discrete demands can be done in a straightforward forward manner, by replacing the binary decision variable $x_k\in\{0,1\}$ by $x_{k,s}$ for $s\in\hat\cS_k$ such that $\sum_{s\in\cS_k}x_{k,s}=1$.}. 
The rest of users, denoted by $\cF \triangleq \cN\backslash\cI$, have continuous demands, defined by convex sets $\cS_k$, for $k\in \cF$; a typical example is  a set defined by box constraints: $\cS_k \triangleq \{s_k \in \CC \mid \underline s_k \le s_k \le \overline s_k \}$, for given lower and upper bounds $\underline s_k$ and $\overline s_k$. 
We consider only {\it consumer}\footnote{On the other hand, when some of the {\it continuous} users are {\it  supplies} (such as in distributed generation), our approximation algorithm for the discrete users with a good approximation guarantee so far cannot be generalized in a straightforward manner (see the footnote in the proof of Lemma~\ref{feas}).} users, such that $\re(s_k)\ge 0$ (but $\im(s_k)$ may be negative) for all $k\in \cN$. This assumption is justified for discrete demands in Sec.~\ref{sec:hard}, as we show fundamental difficulties of OPF when it has discrete control variables for both demand ($\re(s_k)\ge 0$) and supply ($\re(s_k)\le 0$).
Fig.~\ref{fig:o1} shows a pictorial illustration of a single link topology with multiple users. 
\begin{figure}[!h]
	\centering
	\includegraphics[scale=1]{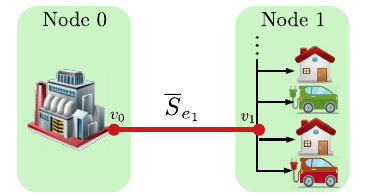}
	\caption{Users declare their power demands to a centralized coordinator, which decides which demand to satisfy subject to operating constraints of the power network.}\vspace*{-10pt}
	\label{fig:o1}
\end{figure}

Let $v_j$ and $\ell_{i,j}$ be the voltage and current magnitude square at node $j$ and edge $(i,j)$, respectively. Let $S_{i,j}$ be the power flowing from node $i$ towards  node $j$. Note that $S_{i,j}$ is not symmetric, namely, $S_{i,j} \ne S_{j,i}$. 

There are several operating constraints of power systems:
\begin{itemize}
\item ({\em Power Capacity Constraints}): \\$|S_{i,j}|\le \overline S_{i,j}, |S_{j,i}|\le \overline S_{i,j},~\forall (i,j) \in \cE$.
\item ({\em Current Thermal Constraints}): $\ell_{i,j}\le \overline \ell_{i,j}, \forall(i,j) \in \cE$.
\item ({\em Voltage Constraints}): $\underline v_j \le v_{j} \le \overline v_j, \forall j \in \cV^+$.      
\end{itemize}
$\underline v_j, \overline v_j \in\RR^+$ are the minimum and maximum allowable voltage magnitude square at  node $j$, and $\overline S_{i,j}, \overline \ell_{i,j}\in \RR^+$ are the maximum allowable apparent power and current on edge $(i,j)$. 
 
The power supply at root $0$ is denoted by $s_0$. This paper adopts the convention to denote a power supply by a complex number with negative real part and a power demand by  a complex number with positive real part (i.e., $ \re(s_0) \le 0$ and  $ \re(s_k)\ge 0$ for all $k\in \cN$).
 In the following, a subscript is omitted from a variable to denote its vector form, such as $S\triangleq(S_{i,j})_{(i,j)\in\cE}, \ell \triangleq (\ell_{i,j})_{(i,j)\in\cE}, s\triangleq(s_k)_{k\in \cN}, x\triangleq(x_k)_{k\in \cI}, v\triangleq(v_j)_{j\in \cV^+}$.

The goal of OPF is to find an assignment for the demand vector  $s$ and supply $s_0$ that minimizes a non-negative  {\it convex} (cost) objective function\footnote{Note we follow the convention in \cite{huang2017sufficient,gan2015exact} that the objective function depends only on the active power; however our procedures and proofs can be extended in a straightforward way to deal with the case when the function $f_k(\cdot)$ depends also on the reactive power. Note also that, due to Eqn.~\ref{p1:con3}, the generation is actually $-s_0$; hence, $f_0$ is {\it non-decreasing} in the amount of generation.} $f$: 
\begin{equation}\label{eq:obj} 
f(s_0, s) = f_0\big(-\re(s_0)\big) + \sum_{k\in\cN }f_k\big(\re(s_k)\big),
\end{equation}
where $f_0$ is the non-negative and {\it non-decreasing} cost for the active power supply (note that $\re(s_0) \le 0$), and $f_k$ is the non-negative and non-increasing cost for each satisfied active power demand, such that $f_k(\re(\overline s_k))=0$ (i.e., each user prefers maximum demand).

We formulate OPF using the branch flow model (BFM)\footnote{To be precise, this model is called branch flow model with angle relaxation \cite{baran1989placement,low2014convex1}, as it omits the phase angles of voltages and currents. But it is always possible to recover the phase angles from a feasible solution in a radial network \cite{low2014convex1}.}.  The inputs are the voltage, current and transmitted power limits $\big[ v_0,(\underline v_j, \overline v_j)_{j\in \cV^+},(\overline S_e, \overline \ell_e, z_e)_{e\in \cE},(\cS_k)_{k\in \cF}, (\hat \cS_k)_{k\in \cI}\big]$, whereas the outputs are the control decision variables and power flow states  $\big[s_0, s, S, v, \ell\big]$.

The OPF that utilizes BFM is given by the following mixed integer programming problem:
\begin{align}
& \text{(OPF)}~ \min_{\substack{s_0, s, S, v, \ell \;\;}} f(s_0, s)  \notag \\
& \text{subject to} \notag \\
 &    \ell_{i,j} =  \frac{|S_{i,j}|^2}{v_i}, &  \forall (i,j) \in \cE,  \label{p1:con1}  \\
& S_{i,j}=  \sum_{k \in \cU_j} s_k  + \sum_{l:(j,l)\in \cE} S_{j,l} + z_{i,j}\ell_{i,j}, &  \forall (i,j) \in \cE,  \label{p1:con2}  \\
& S_{0,1} = - s_0, \label{p1:con3}\\ 
&       v_j = v_i + |z_{i,j}|^2 \ell_{i,j} - 2 \re(z_{i,j}^\ast  S_{i,j}), &   \forall (i,j) \in \cE, \label{p1:con4} \\
& \underline v_j \le v_j \le \overline v_j,    & \forall j \in\cV^+\label{p1:con5}  \\
& |S_{i,j}| \le \overline S_{i,j},~|S_{j,i}| \le \overline S_{i,j}, & \forall (i,j) \in \cE, \label{p1:con6}\\
& \ell_{i,j} \le\overline \ell_{i,j},  & \forall (i,j) \in \cE, \label{p1:con7}\\
&   s_k \in \cS_k & \forall k \in \cF, \label{p1:con8.1}\\
& s_k \in \hat \cS_k, & \forall k \in \cI, \label{p1:con8.2}\\
& v_j \in \RR^+,  \forall j \in \cV^+, ~ \ell_{i,j} \in \RR^+, S_{i,j} \in \CC \label{p1:con9}, &  \forall(i,j)\in \cE.
\end{align}

We make a few remarks:
\begin{enumerate}
\item 
BFM can be also expressed using the opposite orientation towards node 0: $S_{j,i}  =  \sum_{(l,j)\in\cE}\big(  S_{l,j}- z_{l,j} |I_{l,j}|^2\big) - \sum_{k\in\cU_j}s_k$ (see e.g., \cite{huang2017sufficient,gan2015exact}). As shown in \cite{low2014convex1}, there is a bijection between the models of the two orientations, since $S_{j,i} = - S_{i,j} + z_{i,j} |I_{i,j}|^2$ and $I_{i,j} = -I_{j,i}$.
\item 
We explicitly consider capacity constraints (Cons.~\raf{p1:con6}) in both directions, whereas  \cite{huang2017sufficient} implicitly considers only one direction. Note that $|S_{j,i}| \le \overline S_{i,j}$ can be reformulated as $|S_{i,j}- z_{i,j}\ell_{i,j}| \le \overline S_{i,j}$.
Our results can be applied to  bi-directional capacity constraints, which are stronger than that of \cite{huang2017sufficient}. See Sec.~\ref{sec:relax} for a discussion. 
\item 
Cons.~\raf{p1:con8.2} are combinatoric constraints with discrete variables. Although \cite{gan2015exact} also considers the possibility of discrete power injections,  it does not solve the respective optimal solutions. 
\item 
The non-linearity of Cons.~\raf{p1:con1} makes this program non-convex (together with the discrete Cons.~\raf{p1:con8.2}). 
\item 
This paper considers a {\it convex} objective function mainly for  {\it computational efficiency}, although, as in \cite{gan2015exact}, only a non-increasing function is required for the exactness of the convex relaxation. {\color{black} More precisely, except for the last claim in Theorem~\ref{thm:exact} (about the computational efficiency of finding the solution $F'$ from $F''$), the theorem holds for non-convex but non-decreasing objective functions. For simplicity and as in \cite{gan2015exact}, we use the term ``convex relaxation" even when the objective function is not convex, but assuming constraints of the relaxation define a convex region.}

\end{enumerate}

\vspace*{-20pt}
\subsection{Approximation Solutions} \label{sec:approx}

This paper provides an efficient approximation algorithm to solve OPF with combinatoric constraints. Approximation algorithms are a well-studied subject in theoretical computer science \cite{Vazirani10}. In the following, we define some standard terminology for approximation algorithms.

Consider a minimization problem $\cA$ with non-negative objective function $f(\cdot)$, let $F$ be a feasible solution to  $\cA$ and  $F^\star$ be an optimal solution to $\cA$. $f(F) $ denotes the objective value of $F$. Let $\OPT= f(F^\star)$ be the optimal objective value of  $F^\star$. A common definition of approximation solution is $\alpha$-approximation, where $\alpha$ characterizes the approximation ratio between the approximation solution and an optimal solution.

\begin{definition}
        For $\alpha> 1$, an $\alpha$-approximation to minimization problem $\cA$  is a feasible solution $F$ such that $$ f(F)\le \alpha \cdot \OPT.$$
\end{definition}

In particular, {\em polynomial-time approximation scheme} (PTAS) is a $(1+\epsilon)$-approximation algorithm  to a minimization problem, for any $\epsilon>0$.  The running time of a PTAS is polynomial in the input size for every fixed $\epsilon$, but the exponent of the polynomial might depend on $1/\epsilon$.  Namely, PTAS allows a parametrized approximation ratio as the running time.

\vspace*{-10pt}
\subsection{Assumptions}

OPF with combinatoric constraints is hard to solve.  Hence, there are some assumptions to facilitate our solutions:
\begin{itemize}
\item [{\sf A0:}] $f_0(-s_0^{\rm R})$ is  non-decreasing in $-s^{\rm R}_0\in\RR^+$. 
\item [{\sf A1:}]  $z_e \ge 0, \forall e\in\cE$, which naturally holds in distribution networks.

\item [{\sf A2:}] $\underline v_j < v_0 < \overline v_j, \forall j\in \cV^+$, which is also assumed in \cite{huang2017sufficient}. Typically in a distribution network, $v_0$ = 1 (per unit), $\underline v_{j}=(.95)^2$ and $\overline v_{j}=(1.05)^2$; in other words, 5\% deviation from the nominal voltage is allowed. 


\item [{\sf A3:}]       $\re(z_{e}^* \overline{s}_k)\ge 0, \forall k\in \cI, e\in \cE$.
 Intuitively, {\sf A3} requires that the phase angle difference between any $z_e$ and $s_k$ for $k \in \cI$ is at most $\frac{\pi}{2}$. This assumption holds, if the discrete demands do not have large negative reactive power. 
                
\item[{\sf A4:}] $\big|\angle s_k- \angle  s_{k'}\big| \le  \tfrac{\pi}{2}$ for any $k, k' \in \cN$.   Intuitively, {\sf A4} requires that the demands have ``similar'' power factors. {\sf A4} can also be stated as $\re( s^*_k s_{k'})\ge 0$. 

\end{itemize}
Assumptions {\sf A3}  and  {\sf A4} are motivated, from a theoretical point of view, by the inapproximability results which will be presented in Sec.~\ref{sec:hard} (if either assumption does not hold, then the problem cannot be approximated within any polynomial factor unless P=NP). Assumption {\sf A3}  also holds in reasonable practical settings, see e.g., \cite{huang2017sufficient}. As we will see in Sec.~\ref{sec:rot}, by performing an axis rotation, we may assume by {\sf A4} that $s_k\ge 0$. Clearly, under this and assumption {\sf A1}, the reverse power constraint in~\raf{p1:con6} (i.e., $|S_{i,j}- z_{i,j}\ell_{i,j}| \le \overline S_{i,j}$) is implied by the forward power constraint ($|S_{i,j}|\le\overline S_{i,j}$). It will also be observed below that under assumptions {\sf A1},   {\sf A2}  and  {\sf A3}, the voltage upper bounds in \raf{p1:con5} can be dropped.

\section{Review of Convex Relaxation of OPF} \label{sec:relax}

This section presents a brief review of convex relaxation of OPF. The idea of relaxing OPF to a convex optimization problem can significantly improve the solvability of OPF. Convex optimization problems can be solved efficiently by a polynomial-time algorithm. Under certain conditions, the convex relaxation can be shown to obtain the optimal solutions of OPF. 
A second order cone programming (SOCP) relaxation of OPF is obtained by replacing Cons.~\raf{p1:con1} by $\ell_{i,j} \ge \frac{|S_{i,j}|^2}{v_i}$.
For convenience of notation, we rewrite Cons.~\raf{p1:con8.2} by $s_k=  \overline s_k x_k$, where $x_k \in \{0,1\}$ is a control variable. A  relaxation of OPF (denoted by \text{cOPF}) is defined as follows:
\begin{align}
\textsc{(cOPF)}&  \min_{\substack{s_0, s, x, S, v,\ell \;\;}} f(s_0, s),  \notag \\
\text{s.t.} \ &\raf{p1:con2}-\raf{p1:con8.1},\raf{p1:con9},\notag \\
& \ell_{i,j} \ge  \frac{|S_{i,j}|^2}{v_i}, & \forall (i,j) \in \cE, \label{p2:con1} \\
& s_k =  \overline s_k x_k, ~x_k\in\{0,1\}, & \forall k\in\cI. \label{p2:con2}
\end{align}
To see that \raf{p2:con1} is indeed an SOCP constraint, note that Cons.~\raf{p2:con1} can be written as
\begin{align*}
\Bigg\|\left(\begin{array}{c}
2S_{i,j}^{\rm R}\\
2S_{i,j}^{\rm I}\\
\ell_{i,j}-v_i
\end{array}
\right)\Bigg\|_2 \le \ell_{i,j}+v_i
\end{align*}
Notice that \textsc{cOPF} is not completely convex due to the discrete constraints in \raf{p2:con2}. 
In fact, the main contribution of the paper is to provide an approximation scheme for solving \text{cOPF}.

For a given $\hat s\in\CC^n$, we denote by {\sc OPF}$[\hat s]$ (resp., {\sc cOPF}$[\hat s]$ the restriction of \textsc{OPF} (resp., {\sc cOPF}) where we set $s=\hat s$. 
The relaxation \textsc{cOPF}$[\hat s]$  is called (efficiently) {\em exact}, if every optimal solution $F^\star$ of \textsc{cOPF}$[\hat s]$  can be converted to an optimal solution of OPF$[\hat s]$, in a {\it polynomial} number of steps. This definition is adopted from \cite{huang2017sufficient,gan2015exact,low2014convex1,low2014convex2}, but also with an emphasis on {\it efficient computation}.
%

There are several sufficient conditions of exactness\footnote{It should be noted that another sufficient condition for exactness was given in \cite{gan2015exact}, but we will not consider here.}, which are imposed on a given (optimal) solution $F$ of \textsc{cOPF}. The following condition was introduced in \cite{huang2017sufficient}:
\begin{itemize}
	
\item [{\sf C1:}] The solution $F=(s_0,s,v,\ell,S)$ of \textsc{(cOPF)} satisfies the following linear system (in $(\widehat S_{i,j})_{(i,j)\in\cE})$ and $(\widehat v_j)_{j\in\cV^+}$):
\begin{align}
&\widehat S_{i,j} = \sum_{k \in \cU_j} s_k + \sum_{ l:(j,l)\in \cE} \widehat S_{j,l} & \forall (i,j)\in\cE,\label{eq:C1rec}\\
&\widehat v_i - \widehat v_j = 2\re(z_{i,j}^* \widehat S_{i,j}) & \forall (i,j)\in\cE, \label{eqn:C2rec}\\
&\re(z_{h,l}^* \widehat S_{i,j}) \ge 0,  \qquad \forall (i,j) \in\cE, & \forall(h,l)\in \cE_j,\label{eq:rsht}\\
&\widehat v_j \le \overline v_j, & \forall j \in \cV^+\label{eq:vhatt}.
\end{align}

In this paper, we will consider a related condition\footnote{Note that condition {\sf C1} is introduced only for comparison purposes, but is never actually used in the paper.}:
\item [{\sf C2:}]  The solution $F=(s_0,s,v,\ell,S)$ of \textsc{(cOPF)} satisfies
\begin{align}
\sum_{k \in \cN_j}\re(z_{h,l}^*  s_k) \ge 0 \quad \forall j \in \cV^+, (h,l)\in \cE_j \cup \{(i,j)\in\cE \}, \label{eq:exact}
\end{align}
where $\cN_j \triangleq \cup_{j\in\cV_j} \cU_j$ is the set of attached users within subtree $\cT_j$, and $\cE_j \cup \{(i,j)\in\cE \}$ is the set of edges of subtree $\cT_j$ and edges that are connected to node $j$.
	
\end{itemize}

In \cite{huang2017sufficient}, it is shown that {\sf C1} is a sufficient condition for exactness of OPF considering uni-directional power capacity constraints from a leaf to the root\footnote{The sufficient condition in \cite{huang2017sufficient} is stated in a slightly different way, because their problem formulation adopts an opposite flow orientation.}. In order to attain exactness of OPF with bi-directional power capacity constraints, a stronger condition ought to be considered. In addition to \raf{eq:rsht}, it is also required that $\re(z^*_{i,j} \widehat S_{i,j})\ge 0$ which gives~\raf{eq:exact}. Note that by~\raf{eq:exact} and {\sf A2}, and the recursive substitution of $\widehat v_j$ from the root in~\raf{eqn:C2rec}, Cons.~\raf{eq:vhatt} is already satisfied as 
$$\widehat v_j = v_0 - 2 \sum_{e \in \cP_j}\re(z_{e}^* \widehat S_e) \le v_0<\overline v_j,$$
where $\cP_j$ denotes the unique path from root $0$ to node $j$. 
Note also that {\sf A3} implies {\sf C2} when $\cN=\cI$ (as {\sf A3} applies to only discrete demands, whereas {\sf C2} applies to all the demands and edges within a subtree).

The next theorem summarizes the sufficient condition of exactness for the relaxation of OPF. 

\begin{theorem}\label{thm:exact} 
	Let $F''=(s_0'', s', v',\ell'',S'')$ be a feasible solution of {\sc cOPF}$[s']$  satisfying {\sf C2}. Under assumptions~{\sf A0}, {\sf A1}, and {\sf A2}, there is a feasible solution  $F'=(s_0', s', v',\ell',S')$ of {\sc cOPF}$[s']$ that satisfies  $\ell'_{i,j} = \frac{|S_{i,j}'|^2}{v_i'}$ for all  $(i,j)\in \cE,$ and $f(F')\le f(F'')$. Given $F''$, such a solution $F'$ can be found in polynomial time. 
\end{theorem}

The proof uses similar techniques as in  \cite{huang2017sufficient,gan2015exact}, but shows an additional result on the efficient conversion to an optimal solution of {\sc OPF}. 
\ifsupplementary
See the appendix for a proof.
\else
See the supplementary materials for a proof. 
\fi

\vspace*{-10pt}
\section{PTAS for OPF with Discrete Power Demands}\label{sec:ptas}

This section presents a $(1+\epsilon)$-approximation algorithm (PTAS) for OPF. Note that we consider the number of links in the distribution network (i.e., $|\cV^+|=|\cE| = m$) to be a constant, but allow the number of users of discrete demands ($|\cI|$) to be a scalable parameter to the problem. Our PTAS is polynomial in running time with respect to $|\cI|$ or generally~$n$. 

\vspace*{-10pt}
\subsection{Rotational Invariance of OPF}\label{sec:rot}

First, note that OPF is {\em rotational invariant}. That is, if the complex-valued parameters  $(z_e)_{e \in \cE}$ and $(s_k)_{k \in \cN}$ are rotated by the same angle (say $\phi$) and  the objective function $ f(s_0, s)$ is counter-rotated by $\phi$ in $s_0$, then there is a bijection between the rotated OPF and the original unrotated OPF. Define the rotated objective function by $f^\phi(s_0, s) \triangleq  f\big(s_0 e^{-{\bf i}\phi}, s \big)$. 

Formally, rotated OPF is defined as follows:
\begin{align}
&\text{({\sc OPF$^\phi$})} \min_{\substack{s_0, s, S, v, \ell \;\;}}f^\phi(s_0, s)  \notag \\
\text{s.t.} \ &  \raf{p1:con1},\raf{p1:con3},\raf{p1:con5}-\raf{p1:con9} \notag \\
&  S_{i,j}=  \sum_{k \in \cU_j} s_k e^{{\bf i}\phi}  + \sum_{l:(j,l)\in \cE} S_{j,l} + z_{i,j}e^{{\bf i}\phi}\ell_{i,j}, & \forall (i,j) \in \cE  \label{p1:rcon2}  \\
&	v_j = v_i + |z_{i,j}|^2 \ell_{i,j} - 2 \re(z_{i,j}^\ast e^{-{\bf i}\phi}  S_{i,j}),  & \forall (i,j) \in \cE \label{p1:rcon4} 
\end{align}
 Similarly, we define a rotated version of {\sc cOPF} as {\sc cOPF$^\phi$}.

\begin{theorem}\label{thm:rot}
There is a bijection between {\sc OPF$^\phi$} and {\sc OPF}. 
Also, there is a bijection between {\sc cOPF$^\phi$} and {\sc cOPF}.		
\end{theorem}

Theorem~\ref{thm:rot} can be proved by showing that a feasible solution $F=(s_0,s,S,v,\ell)$ of  {\sc OPF$^\phi$} can be mapped to a feasible solution ${\widetilde F}=({\widetilde s}_0, {\widetilde s},{\widetilde S},v,\ell)$ of OPF, where ${\widetilde S}_{i,j}\triangleq S_{i,j} e^{-{\bf i}\phi}, {\widetilde s}_0\triangleq s_0 e^{-{\bf i} \phi}$, and vise versa. Similarly, it holds for {\sc cOPF$^\phi$} and {\sc cOPF}.

\begin{figure}[!htb]\vspace*{-10pt}
	\begin{center}
		\includegraphics[scale=.5]{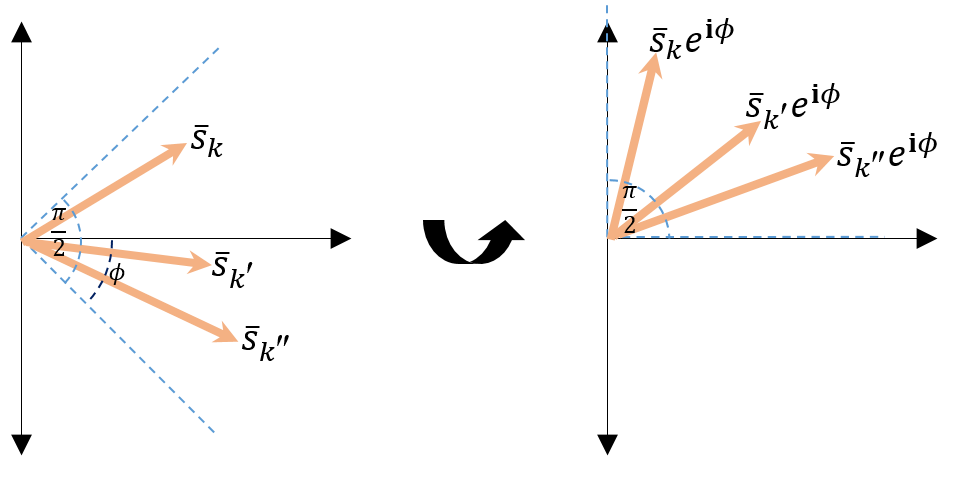}
	\end{center}\vspace{-15pt}
	\caption{Angle $\phi$ is the minimum angle of rotation to rotate $(\overline s_k)_{k \in \cI}$ into the first quadrant.}\vspace*{-15pt}
	\label{fig:rotate}
\end{figure}

\subsection{PTAS Algorithm}
\begin{figure}[!htb]\vspace*{-20pt}
	\begin{center}
		\includegraphics[scale=.38]{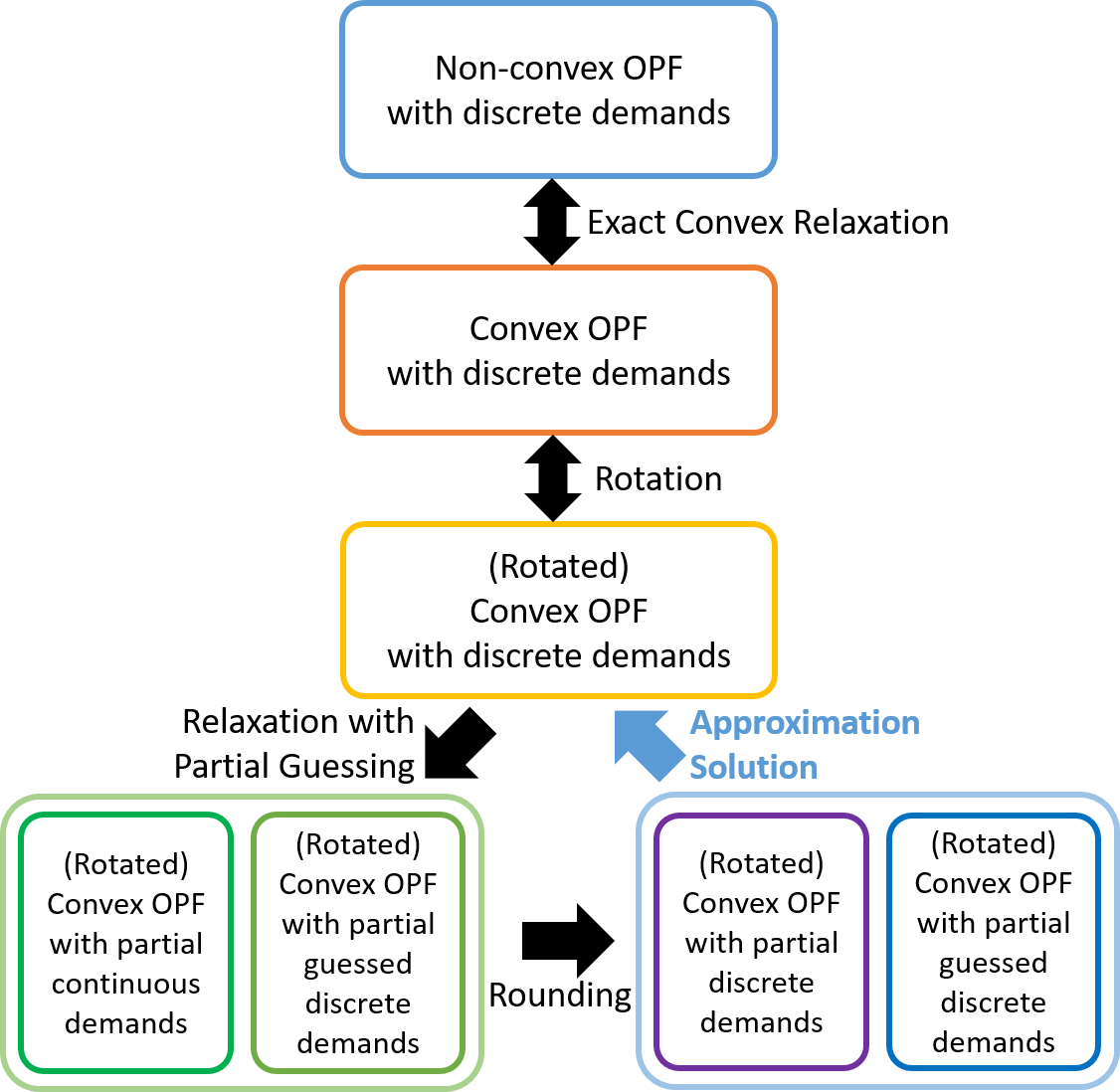}
	\end{center}\vspace*{-5pt}
	\caption{Basic steps of PTAS for OPF.}
	\label{fig:chart}\vspace*{-10pt}
\end{figure} 

Therefore, in the rest of paper, it is more convenient to consider {\sc cOPF$^\phi$}  instead of {\sc cOPF}, where $\phi \triangleq \max\big\{\max_{k\in \cI} \{-\angle \overline s_k\},0\big\}\in[0,\frac{\pi}{2}]$.  Namely, $\phi$ is the minimum angle in order to rotate all the discrete demands from the fourth quadrant to  the first quadrant (see Fig.~\ref{fig:rotate}).
%
%
%
%
Theorem~\ref{thm:rot} allows us to replace assumptions~{\sf A0} and {\sf A4} by the following assumptions:
\begin{itemize}
	\item[{\sf A0$'$}:] $f_0(-s_0^{\rm R} \cos \phi - s_0^{\rm I} \sin \phi)$ is  non-decreasing in $-s_0^{\rm R}, -s_0^{\rm I}$.
	\item [{\sf A4$'$:}] $s_k \ge 0$ for all $k\in \cN$. This is because all demand sets satisfying {\sf A4} are now in the first quadrant after the rotation by $\phi$.
\end{itemize}

Note that assumption  {\sf A1} continues to hold for {\sc OPF$^\phi$}, assuming the original \textsc{OPF} problem satisfies {\sf A3}:  $z_e e^{{\bf i}\phi}\ge 0,  \forall e\in\cE$. This is because of {\sf A3}, namely, $\re(z_{e}^\ast \overline s_k)\ge 0,  \forall k\in \cN,~e\in \cE$, such that the phase angle difference between $z_e$ and $\overline s_k$ is at most $\frac{\pi}{2}$. 
Note also that  {\sf A1}  and {\sf A4$'$} already imply {\sf A3}.

This section presents a PTAS for solving {\sc cOPF$^\phi$}. Together with Theorems~\ref{thm:exact} and \ref{thm:rot}, one can solve OPF by a PTAS.
The basic steps of PTAS are illustrated in Fig.~\ref{fig:chart}. After convex relaxation and rotation, we enumerate possible partial guesses for configuring the control variables of a small subset of  discrete demands. For each guess, we solve the remaining subproblem by relaxing the other discrete control variables to be continuous control variables, and then rounding the continuous control variables to obtain a feasible solution. This algorithm can attain a parameterized approximation ratio by carefully adjusting the number of partial guesses and rounding.


A formal description of the PTAS algorithm (named {\sc PTAS-cOPF})  is presented as follows: 

\begin{enumerate}

\item First, define a {\em partial guess} by $I_1,I_0\subseteq\cI$, such that $I_1 \cap  I_0 = \varnothing$. For each guess, we set $x_k = 1, \forall k \in I_1$ and  $x_k = 0, \forall k \in I_0$. 

\item Define a variant of  {\sc cOPF$^\phi$} with partially pre-configured and partially relaxed discrete control variables, denoted by {\sc P1$[I_0,I_1]$}, as follows:
\begin{align}
&\text{(P1$[I_0,I_1]$)}\quad \min_{\substack{s_0, s, x, S, v, \ell \;\;}} f^\phi(s_0, s)  \notag \\
\text{s.t.} \ & \raf{p1:con3},\raf{p1:con5}-\raf{p1:con8.1},\raf{p1:con9},\raf{p2:con1},\raf{p1:rcon2},\raf{p1:rcon4}, {\sf C2} \notag\\
& s_k = \overline{s}_k x_k, \ \forall k \in \cI, \\
& x_k = 1, \ \forall k \in I_1, \\
&  x_k = 0, \ \forall k \in I_0, \\
& x_k \in [0,1], \ \forall k \in \cI \backslash (I_0 \cup I_1) 
\end{align}
Note that {\sc P1$[I_0,I_1]$} is a convex programming problem and is solvable in polynomial time. Then, obtain an optimal solution of {\sc P1$[I_0,I_1]$}, denoted by $F'=\big(s'_0, s', x', S', v', \ell'\big)$. $F'$ may not satisfy the discrete demand constraints \raf{p1:con8.2} in {\sc cOPF$^\phi$}.  Next, $F'$ will be rounded to obtain a feasible solution to  {\sc cOPF$^\phi$}.

\item Define $\cI'\triangleq \cI \setminus (I_0 \cup I_1)$, and $\overline{f}_k \triangleq f_k(0)$ for $k \in \cI$. Define {\sc P2$[F',\cI']$} as a sub-problem of selecting a subset of discrete control variables ($(x_k)_{k \in \cI'}$) for rounding, based on the respective objective values, as follows\footnote{Note that the part of the objective function corresponding to discrete users can be made linear, without loss of generality, as $f_k(s_k^{\rm R})$ takes only one of two values $\overline f_k=f_k(0)$, if $x_k=0$, or  $0$, otherwise.   }:
\begin{align}
&\qquad \text{(P2$[F',\cI']$)} \quad  \min_{\substack{ x_k\in[0,1], k \in \cI'}} \sum_{k \in \cI} \overline{f}_k( 1- x_k)  \notag \\
 &\text{s.t.\ } 0 \le \sum_{k \in \cN}  \re\Big( \sum_{(h,l)\in \cP_k \cap \cP_j} z^*_{h,l} {s}_k\Big)   \label{r2:con0} \\
  & \qquad  \le\sum_{k \in \cN}  \re\Big( \sum_{(h,l)\in \cP_k \cap \cP_j} z^*_{h,l} s'_k\Big), \forall j \in \cV^+ \label{r2:con1}\\
&\sum_{k \in \cN_j} \re({s}_k e^{{\bf i}\phi})  \le \sum_{k \in \cN_j } \re(s'_k e^{{\bf i}\phi}), \forall j \in \cV^+ \label{r2:con2}\\
&\sum_{k \in \cN_j} \im ({s}_k e^{{\bf i}\phi})  \le  \sum_{k \in \cN_j } \im(s'_k e^{{\bf i}\phi}), \forall j \in \cV^+ \label{r2:con3} \\
& s_k = \overline s_k  x_k, \ \forall k\in\cI', \quad  s_k = s'_k, \ \forall \cN \backslash \cI' \label{r3:con4}
\end{align}
Note that {\sc P2$[F',\cI']$} is a linear programming problem.

\item Suppose $(x''_k)_{k\in \cI'}$ is an optimal solution of {\sc P2$[F',\cI']$}. Each $x''_k$ is rounded to an integral solution such that
\begin{equation}
	\hat x_k = \left\{\begin{array}{ll}
	\lfloor x''_k \rfloor, & \mbox{if\ } k \in \cI',\\
	    x'_k,  &   \mbox{if\ } k\in \cI\setminus \cI'
	\end{array}\right.
\end{equation}

\item Then, obtain the corresponding $\hat s_0,\hat s,\hat S,\hat\ell,\hat v$ by \text{P3$[\hat x,s']$}:
\begin{align}
\text{(P3$[\hat x,s'$])}& \min_{\substack{\hat s_0, \hat s, x, \hat S, \hat v, \hat \ell \;\;}} f^\phi(s_0, s) \notag &\\
\text{s.t.} \ & \raf{p1:con3},\raf{p1:con5}-\raf{p1:con7},\raf{p1:con9},\raf{p2:con1},\raf{p1:rcon2},\raf{p1:rcon4},{\sf C2} \notag\\
&s_k = s'_k, \ \forall k\in \cF, \\
&s_k = \overline s_k \hat x_k, \ \forall k\in \cI
\end{align}
Note that {\sc P3$[\hat x,s'$]} is a convex programming problem.

\item The output solution will be the one having the maximal objective value among all guesses.

\end{enumerate}
The pseudo-codes of  {\sc PTAS-cOPF} are given in  Algorithm~\ref{alg:ptas}. 

\begin{algorithm}[!htb]
\caption{ {\sc PTAS-cOPF} \label{alg:ptas} }
\begin{algorithmic}[1]
\Require $\epsilon, v_0,(\underline v_j, \overline v_j)_{j\in \cV^+}, (\overline S_e, \overline  \ell_e, z_e)_{e\in \cE}, (\underline s_k, \overline s_k)_{k\in \cN}$
\Ensure	Solution $\hat F=(\hat s_0, \hat s, \hat x, \hat S, \hat v, \hat \ell)$ to {\sc cOPF$^\phi$}
\State $f_{\min} \leftarrow \infty$
\For {each set  $I_0\subseteq \cI$ such that $|I_0| \le \frac{4m}{\epsilon}$ } \label{alg:guess} 
\State $I_1\leftarrow\big\{k\in\cI\setminus I_0 \mid \overline f_{k} >\min_{k'\in I_0}\{\overline f_{k'}\}\big\}$ \label{alg:mcv}

\State $\cI' \leftarrow \cI \backslash (I_0 \cup I_1)$
\If{{\sc P1$[I_0, I_1]$} is feasible} \label{alg:mc-feas}

\State $F' \leftarrow$  Optimal solution of {\sc P1$[I_0, I_1]$} \label{alg:cvx}
\State $(x''_k)_{k\in \cI'} \leftarrow$ Optimal solution of {\sc P2}$[F',\cI']$ \label{alg:mclp}
\State$ (\hat x_k)_{k\in \cN} \leftarrow  \big( (\lfloor x''_k \rfloor)_{k\in \cI'}, ( x'_k )_{k\in \cN\setminus\cI'} \big)$  \label{alg:mc-round}
\State $(\hat s_0, \hat s, \hat S, \hat v, \hat \ell)\leftarrow$  Optimal solution of {\sc P3$[\hat x,s']$}  \label{alg:r3}
\If{$f_{\min} >  f^\phi(\hat s_0, \hat s)$ }
\State $ \hat F \leftarrow (\hat s_0, \hat s, \hat x, \hat S, \hat v, \hat \ell), \ f_{\min} \leftarrow  f^\phi(\hat s_0, \hat s)$
\EndIf 
\EndIf

\EndFor
\State \Return $\hat F$
\end{algorithmic}

\end{algorithm}

\vspace*{-10pt}
\subsection{Proving Approximation Ratio}

In this section, the approximation ratio of PTAS will be derived as $(1+\epsilon)$, if one sets the size of partial guesses of satisfiable discrete demands by $|I_0|\le \frac{4m}{\epsilon}$, where $m$ is the number of nodes in distribution networks, and the corresponding $I_0$ by
\begin{equation} \label{cond:I1}
I_1=\big\{k\in\cI\setminus I_0 \mid \overline f_{k} >\min_{k'\in I_0}\{\overline f_{k'}\}\big\}
\end{equation}
Therefore, one can adjust the approximation ratio by limiting the size of $I_0$ in partial guessing.

\begin{remark} \label{r1}{The running time of Algorithm~\ref{alg:ptas} is $O(n^{4m/\epsilon}T)$, where $T$ is the time required to solve the convex programs P1, P2 and P3, (the space requirement is polynomial in $n$).} To speed up {\sc PTAS-cOPF}, one can first compute the optimal objective value (denoted by $\underline{f}$) of {\sc P1} by taking $I_0 = I_1 = \varnothing$, which naturally is a lower bound to that of {\sc cOPF}.  {\sc PTAS-cOPF} will stop and return a solution, if the gap between the solution's objective value and $\underline{f}$ is sufficiently small. Hence, this may skip the partial guessing, if $\underline{f}$ is already closed to the solution of {\sc PTAS-cOPF} without partial guessing (which is often observed in the evaluation in Sec.~\ref{sec:sim}).  
\end{remark}
\begin{theorem}\label{th:opftas} 
With assumptions {\sf A0',A1,A2,A3,A4'}, for any given $\epsilon > 0$, {\sc PTAS-cOPF} provides a $(1+\epsilon)$-approximation solution of {\sc cOPF$^\phi$} in
{\color{black} $O(n^{4m/\epsilon}T)$, where $T$ is the time required to solve the convex programs P1, P2 and P3, assuming $m$ and $\epsilon$ are fixed constants.}
\end{theorem}

\begin{proof}
As mentioned in Remark~\ref{r1}, the running time of {\sc PTAS-cOPF} is polynomial in $n$, for fixed $\epsilon$ and $m$. Next, we show that the output solution $\hat F$ is $(1+ \epsilon)$-approximation for {\sc cOPF$^\phi$}. 

Let $F^\star =  (s_0^\star, s^\star, x^\star, S^\star, v^\star, \ell^\star)$ be an optimal solution of {\sc cOPF$^\phi$}. Define  $I^\star_0\triangleq\{k \in \cI \mid x^\star_k = 0\}$ and $I^\star_1\triangleq\{k \in \cI \mid x^\star_k = 1\}$ be the sets of unsatisfiable and satisfiable discrete demands in $F^\star$, respectively.
There are two cases:
\begin{enumerate}

\item If $|I^\star_0|\le\tfrac{4m}{\epsilon}$, then there exists a partial guess $I_0$, such that $I_0 = I^\star_0$. Namely, {\sc PTAS-cOPF} can find $F^\star$ by enumerating all possible  $I_1$ such that $|I_0|\le \frac{4m}{\epsilon}$. Therefore, it produces the optimal solution of {\sc cOPF$^\phi$}.

\item Otherwise, $|I^\star_0|>\tfrac{4m}{\epsilon}$, then  {\sc PTAS-cOPF} can still find some $I_0$, which is a subset of satisfiable discrete demands in $I^\star_0 $ with a number of $\lfloor\tfrac{4m}{\epsilon}\rfloor$ highest ${\overline f}_k$:
\begin{equation} \label{cond:I0}
I_0 \subseteq I^\star_0  \mbox{\ and\ }  |I_0| = \lfloor\tfrac{4m}{\epsilon}\rfloor  \mbox{\ and\ }  \min_{k \in I_0} \{{\overline f}_k\} \ge \max_{k \in  I^\star_0 \backslash I_0} \{{\overline f}_k\} 
\end{equation}
Next, we assume $I_0$ satisfying \raf{cond:I0} and $I_1$ satisfying \raf{cond:I1}. Note that $I_1\cap I^\star_0 = \varnothing$ (and hence $I_1 \subseteq I^\star_1$). Otherwise, $I_0\cap I_1 \ne \varnothing$, because of \raf{cond:I0}.
\end{enumerate}

Then, we focus on case 2. Let us consider $F'  = (s_0', s', x', S', v', \ell')$, which is an optimal solution of {\sc P1$[I_0,I_1]$}.  Since $I_0 \subseteq I^\star_0$ and $I_1 \subseteq I^\star_1$, it follows that
\begin{equation}  \label{eq:mc-cvx}
f^\phi(s'_0, s') \le  f^\phi(s_0^\star, s^\star)
\end{equation}

Next, let us consider $(x''_k)_{k\in \cI}$, which is an optimal solution of {\sc P2}$[F',\cI']$. Note that $(x'_k)_{k\in \cI}$ is also a feasible solution to {\sc P2}$[F',\cI']$ (where Cons.~\raf{r2:con1}-\raf{r2:con3} are tight, and $s'$ satisfies {\sf C2}, and hence, \raf{r2:con0}). It follows that 
\begin{equation}  \label{eq:mc-lp}
\sum_{k \in \cI}\overline f_k (1-x''_k)   \le \sum_{k \in \cI}\overline f_k (1-x'_k)
\end{equation}

Note that {\sc P2}$[F',\cI']$ is a linear programming problem. By Lemma~\ref{lem:lp-bfs}, at most  $4m$ components in $(x''_k)_{k\in \cI'}$ are fractional. For each fractional component, say  $k\in\cI'$ , one obtains $\overline f_{k} \le \min_{k' \in I_0}\{ \overline f_{k'}\}$. Otherwise, $k \in I_1$ by  \raf{cond:I1}. Hence,
\begin{equation}\label{eq:fk_I0}
\overline f_{k} \le \min_{k' \in I_0}\{ \overline f_{k'}\} \le \frac{1}{|I_0|} \sum_{k' \in I_0} \overline f_{k'}
\end{equation}	
Therefore, rounding down $x''$ for at most  $4m$ components (i.e., $\hat x_k = \lfloor x''_k \rfloor$ for $k \in \cI'$), by \raf{eq:fk_I0} one obtains
\begin{align}
  \sum_{k\in \cI} \overline f_k(1- x''_k)  \ge & \sum_{k\in \cI} \overline f_k (1-\hat x_k) - \sum_{k\in \cI:0<x''_k<1} \overline f_k \notag  \\
  \ge &  \sum_{k\in \cI} \overline f_k (1-\hat x_k) - \frac{4m}{|I_0|} \sum_{k \in I_0} \overline f_{k} \notag \\
  \ge &  \sum_{k\in \cI} \overline f_k (1-\hat x_k) - \frac{4m}{|I_0|} \sum_{k \in \cI} \overline f_{k} (1-x''_k) \notag \\
  \ge &  \sum_{k\in \cI} \overline f_k (1-\hat x_k) - \epsilon \sum_{k \in \cI} \overline f_{k} (1-x''_k)  \label{eq:hatlp}
\end{align}	

By \raf{eq:mc-lp}, \raf{eq:hatlp} and non-negativity of $f_0, f_k$, it follows that
\begin{align}
  \sum_{k\in \cI} \overline f_k(1-\hat x_k) \le &  (1+\epsilon) \sum_{k \in \cI}\overline f_k(1-x'_k) \notag \\
  \le  & \sum_{k \in \cI}\overline f_k (1-x'_k)  + \epsilon f^\phi(s_0', s') \label{eq:1-e}
\end{align}

Finally, by \raf{eq:1-e} and Lemma~\ref{feas} below, one obtains
$$
f^\phi(\hat s_0, \hat s) \le  (1+\epsilon) f^\phi(s_0', s') \le(1+\epsilon) f^\phi(s_0^\star, s^\star)$$
Hence, this completes the proof.
\end{proof}



\vspace*{-10pt}
\begin{lemma}[see \cite{lpbook}] \label{lem:lp-bfs}
Consider linear programming problem:
\begin{align}
\min_{{\bf x} \in [0, 1]^n} \; & {\bf c}^T \cdot {\bf x}\\
\mbox{s.t.} \;\;& {\bf A} {\bf x} \le {\bf b} 
\end{align}
where ${\bf A}$ is an $n\times r$ matrix. Then there exists an optimal solution ${\bf x}^\star$ {\color{black} that can be computed in polynomial time} such that at most $r$ components are fractional. Namely, $|\{ i = 1, ..., n\} \mid x^\star_i \in (0,1)| \le r$.
\end{lemma}

Proof of Lemma~\ref{lem:lp-bfs} follows from the properties of basic feasible solutions of linear programming problems.

\begin{lemma}\label{feas}
Let $F'=\big(s'_0, s', x',S',v',\ell'\big)$ be a feasible~solution of {\sc P1$[I_0,I_1]$}.
If $\hat x_k\in[0,1], \forall k \in \cI$ satisfies the following:
\begin{align}
&\sum_{k\in\cI}\overline f_k(1-\hat x_k)\le \sum_{k\in\cI}\overline f_k(1-x_k')+\epsilon f^\phi(s'_0, s'), \label{feas:con0}\\
&0 \le \sum_{k \in \cN}  \re\Big( \sum_{(h,l)\in \cP_k \cap \cP_j} z^*_{h,l} \hat s_k\Big) \label{feas:con1.1} \\
& ~ \le \sum_{k \in \cN}  \re\Big( \sum_{(h,l)\in \cP_k \cap \cP_j} z^*_{h,l} s'_k\Big),\ \forall j \in \cV^+  \label{feas:con1}\\
&\sum_{k \in \cN_j} \re(\hat s_k e^{{\bf i}\phi})  \le \sum_{k \in \cN_j} \re(s'_k e^{{\bf i}\phi}),\ \forall  j \in \cV^+ \label{feas:con2}\\
&\sum_{k \in \cN_j} \im(\hat s_k e^{{\bf i}\phi})  \le\sum_{k \in \cN_j} \im(s'_k e^{{\bf i}\phi}), \ \forall j \in \cV^+ \label{feas:con3}\\
&\hat s_k = \overline s_k \hat x_k, \ \forall k\in\cI, \quad \hat s_k = s'_k, \ \forall k\in\cF \\
& \hat x_k= x'_k, \ \forall k\in\cI \backslash \cI' \label{feas:con4}
\end{align}
Then, with assumptions {\sf A0',A1,A2,A3,A4'}, there exists a feasible solution $\hat F=\big(\hat s_0, \hat s, \hat x,  \hat  S, \hat v, \hat \ell_e\big)$ of \textsc{P3$[\hat x,s']$}, such that $f^\phi(\hat s_0, \hat s) \le  (1+\epsilon) f^\phi(s_0', s')$.
\end{lemma}

\noindent
\begin{proof}
	First, we aim to construct a feasible solution $\hat F=(\hat s_0, \hat s, \hat x, \hat S, \hat v, \hat \ell)$ for {\sc P3$[\hat x,s']$}.
	Reformulate Cons.~\raf{p1:rcon4} by recursively substituting $v_j'$, from the root to node $j$, and then recursively substituting $S'_{h,l}$, for each $(h,l)$ on the path from $j$ to the root:
	\vspace*{-10pt}
	\begin{align}
	v_j &= v_0 -2\sum_{ (h,l) \in \cP_{j} } \re(z_{h,l}^*  e^{-{\bf i}\phi} S_{h,l}) +  \sum_{( h,l) \in \cP_{j} } |z_{h,l}|^2  \ell_{h,l} \notag\\
	&=v_0 -2 \sum_{ (h,l) \in \cP_{j} } \re\Big(z^*_{h,l} e^{-{\bf i}\phi}\big(\sum_{k\in \cN_l}  s_k e^{{\bf i}\phi}  + \sum_{e \in \cE_l} z_{e} e^{{\bf i}\phi} \ell_{e} \big)\Big) \notag \\ 
	&~~+   \sum_{ (h,l) \in \cP_{j} } |z_{h,l}|^2  \ell_{h,l} \notag\\
	&=v_0 -2 \sum_{k \in \cN} \re\Big( \sum_{(h,l)\in \cP_k \cap \cP_j} z^*_{h,l} s_k\Big)    -    2\sum_{ (h,l) \in\cP_{j} }|z_{h,l}|^2  \ell_{h,l}  \notag \\ 
	&~~- 2 \sum_{(h,l)\in \cP_j} \re \Big( z_{h,l}^* \sum_{e \in \cE_l} z_{e}  \ell_{e}\Big) +\sum_{ (h,l) \in\cP_{j} }|z_{h,l}|^2  \ell_{h,l}, \notag
	\end{align}
	where the last equality follows from exchanging the summation operators, and $z_{e}^* z_{e} = |z_{e}|^2$.
	Hence, one obtains
	\begin{align}
	&v_j=v_0 -2 \sum_{k \in \cN} \re\Big( \sum_{(h,l)\in \cP_k \cap \cP_j} z^*_{h,l} s_k\Big)  \notag \\
	&~~    - \Big(2 \sum_{(h,l)\in \cP_j} \re \big( z_{h,l}^* \sum_{e \in  \cE_l} z_{e}  \ell_{e}\big) +    \sum_{ (h,l) \in\cP_{j} }|z_{h,l}|^2  \ell_{h,l}\Big)    \label{eq:rec-v}
	\end{align}
	
	Set a feasible solution $\hat F$ as follows:
	\begin{align}
	\hat \ell_{i,j} & = \ell'_{i,j},\\
	\hat S_{i,j} & = \sum_{ k\in\cN_j} \hat s_k e^{{\bf i}\phi} + \sum_{e \in \cE_i}z_{e} e^{{\bf i}\phi}\hat \ell_{e},\\
	\hat v_j & = v_0 - 2 \sum_{k \in \cN} \re\Big( \sum_{(h,l)\in \cP_k \cap \cP_j} z^*_{h,l} \hat s_k\Big)  + \hat{L}_{i,j},
	\end{align}
	where 
	\begin{align}
	& \hat L_{i,j} \triangleq  \notag \\ 
	& - \Big(2 \sum_{(h,l)\in \cP_j} \re \big( z_{h,l}^* \sum_{e \in \cE_l} z_{e} \hat{\ell}_{e}\big) +    \sum_{ (h,l) \in\cP_{j} }|z_{h,l}|^2 \hat{\ell}_{h,l}\Big) \notag
	\end{align}
	Rewriting  $S'_{i,j}$ recursively, one obtains
	\begin{align}
	S'_{i,j}=\sum_{k \in \cN_j} s'_k  e^{{\bf i}\phi} + \sum_{e \in \cE_i}z_e e^{{\bf i}\phi} \ell'_{e}
	\end{align}
	Because of  \raf{feas:con2} and \raf{feas:con3}, it follows that
	\begin{align}
	\hat S_{i,j} \le  \sum_{ k\in\cN_j} s'_k e^{{\bf i}\phi} + \sum_{e \in \cE_i}z_e e^{{\bf i}\phi}  \ell'_{e}  = S'_{i,j} \label{eq:tc'}
	\end{align}
	
	The feasibility of $\hat S_{i,j}$ can be shown as follows: By {\sf A1, A4'}, it follows that\footnote{Note that we need here that all users are consumers.} $\hat S_{i,j} \ge 0$. By \raf{eq:tc'}, one obtains
	\begin{equation}\label{eq:bS}
	|\hat{S}_{i,j}| \le |S'_{i,j}| \le \overline S_{i,j}
	\end{equation}
	
	The feasibility of $\hat v_j$ can be shown as follows:
	By~\raf{feas:con1}, one obtains
	\begin{align}
	\hat v_j  &\ge  v_0 -2 \sum_{k \in \cN} \re\Big( \sum_{(h,l)\in \cP_k \cap \cP_j} z^*_{h,l} s'_k\Big)  + \hat L_{i,j} = v'_j \ge \underline v_j, \label{eq:bV}
	\end{align}
	where the last inequality follows from the feasibility of $v'_j$. Since
	$\hat L_{i,j}\le 0$ (by {\sf A1}), {\sf A2} and \raf{feas:con1.1}, one obtains
	$$\hat v_j \le  v_0 -2 \sum_{k \in \cN} \re\Big( \sum_{(h,l)\in \cP_k \cap \cP_j} z^*_{h,l} \hat s_k\Big) \le v_0 <\overline v_j.$$
	
	The feasibility of $\hat \ell_{i,j}$ can be shown as follows:
	By \raf{eq:bS}, \raf{eq:bV}, one obtains $\hat \ell_{i,j} = \ell_{i,j}' = \frac{|S'_{i,j}|^2}{v'_i} \ge \frac{|\hat S_{i,j}|^2}{\hat v_i}$, hence $\hat \ell_{i,j}$ satisfies Cons.~\raf{p2:con1}. 
	
	Finally, one obtains
	\begin{align}\label{eq:increasing}
	\hat s_0 = - \hat S_{0,1}  \ge - S'_{0,1} = s'_0
	\end{align}
	Because {\sf A1, A4'}, $e^{-{\bf i}\phi}\hat s_0$ and  $e^{-{\bf i}\phi}s'_0$ are in the third quadrant, which is the opposite to $s_k e^{{\bf i}\phi}$ and $ z_e e^{{\bf i}\phi}$ (by \raf{r2:con1},\raf{p1:con3}). Hence,  because of \raf{feas:con0}, \raf{feas:con4} and {\sf A0'} ($f_0(\re(\cdot))$ is non-increasing for parameters in the third quadrant), it follows that
	\begin{align}
	& \sum_{k\in\cI}\overline f_k(1-\hat x_k) + \sum_{k\in\cF} f_k\big(\re( \hat s_k)\big) + f_0\big(\re( e^{-{\bf i}\phi}\hat s_0)\big)  \notag \\
	\le \ & \sum_{k\in\cI}\overline f_k(1-x_k') + \sum_{k\in\cF} f_k\big(\re( s'_k)\big) + f_0\big(\re( e^{-{\bf i}\phi} s'_0)\big ) \notag \\
	& +\epsilon f^\phi(s'_0, s') \\
	\Rightarrow \ & f^\phi(\hat s_0, \hat s) \le  (1+\epsilon) f^\phi(s'_0, s')
	\end{align}
	This completes the proof.
\end{proof}
\vspace*{-5pt}
\begin{remark}
	Intuitively speaking, Theorem~\ref{thm:exact} states that, under the mentioned assumptions, if the relaxed (discrete) problem $\textsc{cOPF}$ has a feasible solution then the original \textsc{OPF} problem has a feasible solution with a smaller or equal objective value. Then step~\ref{alg:mc-feas} of Algorithm~\ref{alg:ptas} checks if the convex relaxation of $\textsc{cOPF}$ indeed has a feasible solution; in such a case, Lemma~\ref{feas} implies that \textsc{OPF} has also a feasible solution.
\end{remark}

\vspace*{-10pt}
\section{Discussion} \label{sec:hard}

While exact convex relaxation of OPF applies to the setting with discrete demands, the efficiency of solving OPF with discrete demands is substantially more challenging than that with only  continuous demands. In prior work  \cite{YC13CKP,khonji2016optimal,chau2016truthful,KCE2014Hard}, we show some fundamental hardness results for OPF with discrete demands. Although the results in \cite{YC13CKP,khonji2016optimal,chau2016truthful,KCE2014Hard} are proven by a slightly different model of maximizing an objective function, with minor modifications to the proofs these results can also be applied to our model of minimizing an objective function.

This paper provides a PTAS for solving OPF with discrete demands. A better alternative to PTAS is a {\em fully polynomial-time approximation scheme} (FPTAS), which requires the running time to be polynomial in both input size $n$ and $1/\epsilon$.

\begin{theorem}[see \cite{YC13CKP}] \label{thm:h0}
Unless {\sc P}={\sc NP}, there exists no {\sc FPTAS} for {\sc OPF} with discrete demands, even for a single-link distribution network $|\cE|=1$  with $z_e = 0$.
\end{theorem}

By Theorem~\ref{thm:h0}, a PTAS is  among the best achievable efficient algorithm that can be attained for OPF with discrete demands, because FPTAS is not possible.

Next, we show that assumption {\sf A3} is necessary for PTAS. 

\begin{theorem}[see \cite{khonji2016optimal}] \label{thm:h1}
Unless {\sc P}={\sc NP}, there exists no $\alpha$-approximation for {\sc OPF} by a polynomial-time algorithm in $n$, for any $\alpha$ that has polynomial length in $n$, if $\re(z^\ast_e s_k)$ is allowed to be arbitrary for any $k\in\cI$,  even for a single-link distribution network $|\cE|=1$.      
\end{theorem}

In  Theorem~\ref{thm:h1}, $\alpha$ can be as large as $2^{P(n)}$, where $P(n)$ is an arbitrary polynomial in $n$.
Thereom~\ref{thm:h1} shows that {\sf A3} is necessary for any approximation algorithm for {\sc OPF}  with discrete demands  to have a practical approximation ratio.

Finally, we show that assumption {\sf A4} is necessary for PTAS. Let $\theta$ be the maximum angle difference between any pair of demands, $\theta \triangleq \max_{k, k' \in \cN}\big|\angle s_k - \angle s_{k'}\big|$. Theorem~\ref{thm:h3} shows that the approximability of {\sc OPF}  with discrete demands depends on $\theta$. Hence, a PTAS requires {\sf A4}.

\begin{theorem}[see \cite{chau2016truthful,KCE2014Hard}] \label{thm:h3}
Unless {\sc P}={\sc NP}, for any $\delta>0$, there is no $\alpha$-approximation for \textsc{OPF} by a polynomial-time algorithm in $n$,  for $\theta \in  [\frac{\pi}{2}+\delta,\pi]$, where $\alpha,\delta$ have polynomial length in $n$ and $\delta$ is exponentially small in $n$, even for a single-link distribution network $|\cE|=1$ with $z_e = 0$.  
\end{theorem}

\vspace*{-10pt}
\section{Evaluation Studies} \label{sec:sim}

In this section, the performance of {\sc PTAS} is evaluated by simulations in terms of optimality and running time. 

\vspace{-10pt}
\subsection{Simulation Settings}

The evaluation was performed on the Bus 4 distribution system of the Roy Billinton Test System (RBTS)\cite{allan1991reliability, huang2017sufficient}, which comprises of 13 nodes, in which the generation source is attached to the sub-station node 0, the base power capacity of this network is  $8$MVA, and the base voltage is $11$KV.

The evaluation was also performed on the IEEE 123-node network, in which the generation source is attached to node 150, and the base capacity and voltage are $5$MVA and $4.16$KV, respectively.

The IEEE 123-node network is an unbalanced three-phase network with several devices that are not modeled in {\sc OPF}. As in \cite{gan2015exact}, one can modify the IEEE network by the following:
\begin{itemize}
	\item The three phases are assumed to be decoupled into three identical single phase networks.
	\item Closed circuit switches are modeled as shorted lines and ignore open circuit switches.
	\item Transformers are modeled as lines with appropriate impedances.
\end{itemize}

\begin{figure*}[!htb]
	\begin{subfigure}[b]{0.5\textwidth}
		\begin{center}
			\includegraphics[scale=0.5]{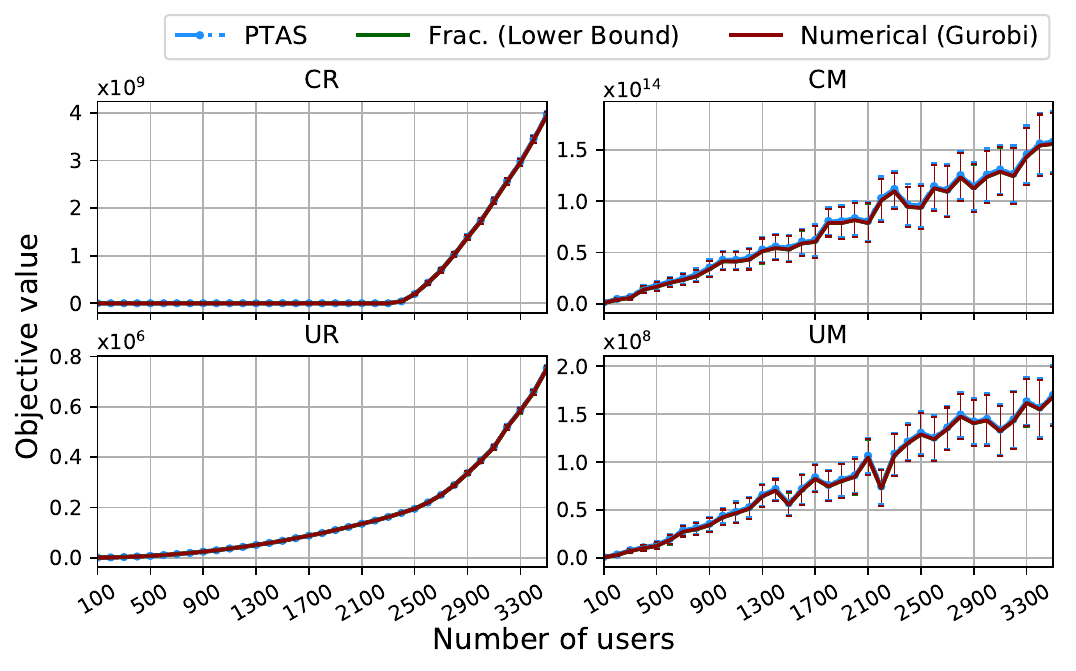}
		\end{center}\vspace*{-10pt}
		\caption{RBTS 13-node network}
		\label{fig:obj13}
	\end{subfigure}
	\begin{subfigure}[b]{0.5\textwidth}
		\begin{center}
			\includegraphics[scale=0.5]{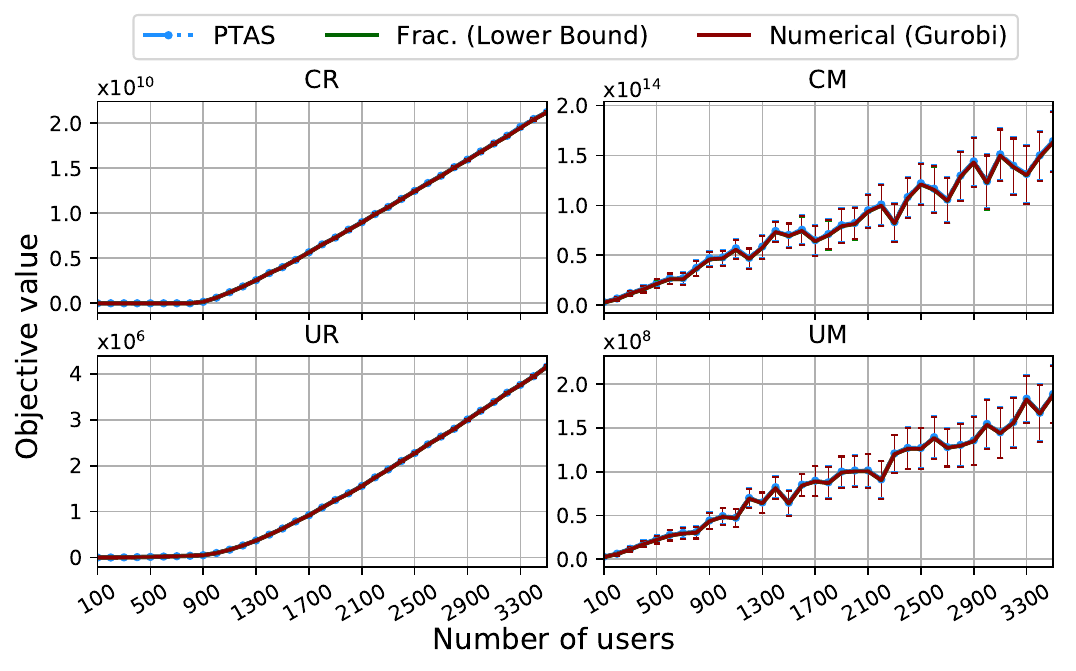}
		\end{center} \vspace*{-10pt}
		\caption{IEEE 123-node network}
		\label{fig:obj123}
	\end{subfigure}  \vspace*{-15pt}
	\caption{The average objective values of {\sc PTAS} (without partial guessing), Gurobi numerical solver, and fractional solutions with relaxed discrete demands (as the lower bounds to the true optimal values), against the number of users with 95\% confidence intervals. }\vspace*{-15pt}
\end{figure*}
\begin{figure}[!h]

	\begin{center}
			\hspace*{-15pt}\includegraphics[scale=0.5]{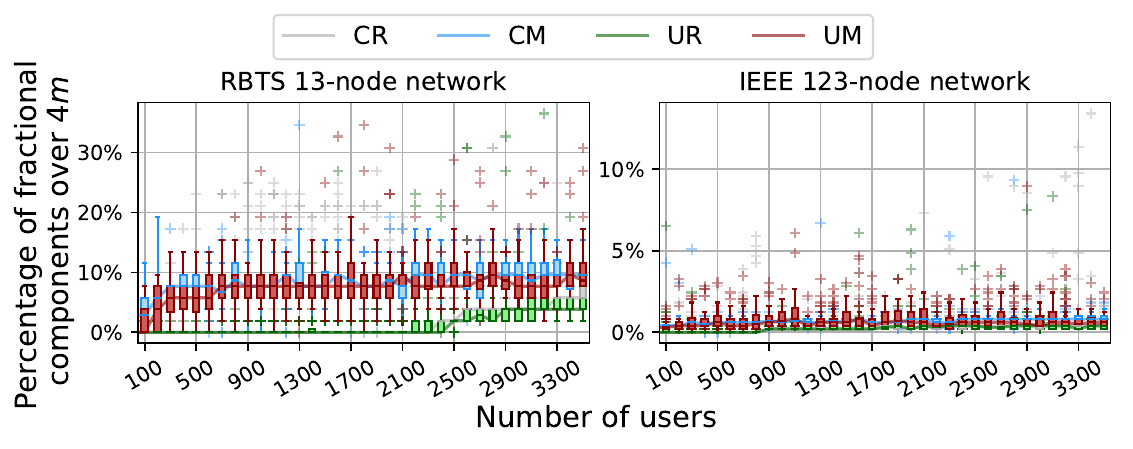}
	\end{center}	\vspace*{-10pt}
	\caption{The ratio of fractional components after solving {\sc P1}.  The `+' points represent the outliers.}
	\label{fig:frac} \vspace*{-15pt}
\end{figure} 
\begin{figure}[!h]
	\begin{center}

		\includegraphics[scale=0.5]{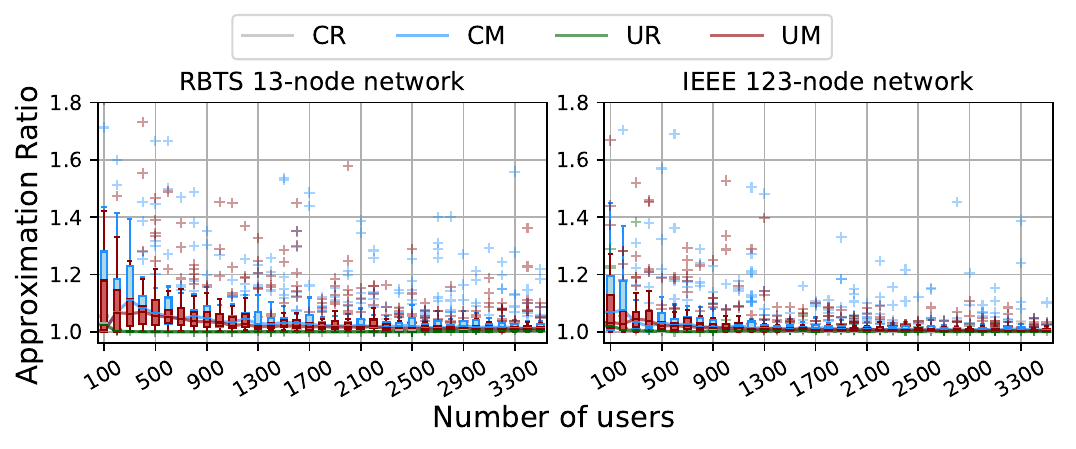}
	\end{center}\vspace*{-10pt}
	\caption{The empirical approximation ratios of PTAS   (without partial guessing) for different case studies, against the number of users.}
	\label{fig:ar}\vspace*{-15pt}
\end{figure}

\begin{figure}[!h]\vspace*{-10pt}
	\begin{center}
		\includegraphics[scale=0.5]{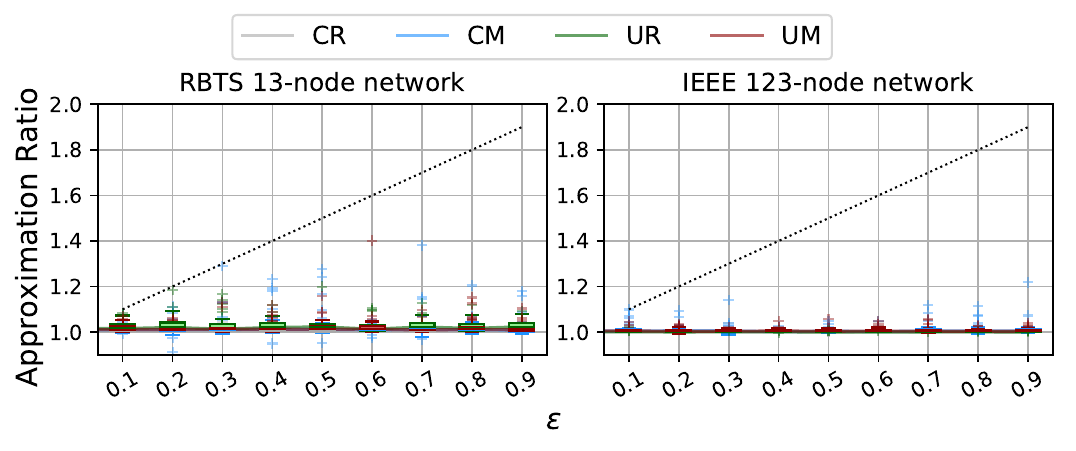}
	\end{center}\vspace*{-10pt}
	\caption{The empirical approximation ratios of PTAS  for different case studies, against different values of $\epsilon$. The diagonal line represents the theoretical bound on the approximation ratio for PTAS.}
	\label{fig:ar_e}\vspace*{-15pt}
\end{figure}

All power demands are discrete and are randomly positioned at the nodes in $\cV^+$ uniformly. Several case studies are considered by different user types and the correlations between user demands and  $f_k$:

\begin{itemize}
\item {\em User Types}:
\begin{enumerate}
	
\item {\em Residential} (R): The users have small power demands ranging from $500$VA to $5$KVA.

\item {\em Industrial} (I): The users have big demands ranging from $300$KVA to $1$MVA with non-negative reactive power.

\item {\em Mixed} (M): The users consist of a mix of industrial and residential users, with less than $20$\% industrial users.

\end{enumerate}
				
\item {\em Cost-Demand Correlation}:
\begin{enumerate}
	
\item {\em Correlated Setting} (C): The cost objective of each user  is a function of his power demand as follows:
$f_k({\re(s_k)}) = {\Big(|\overline s_k| - \tfrac{\re(s_k)}{\re(\overline s_k)} |\overline s_k|\Big)  }^2.$ 

\item {\em Uncorrelated Setting} (U): 
The cost objective of each user is independent of his/her power demand and is generated randomly from $[0, |\overline s_{\max}(k)|]$. Here $\overline s_{\max}(k)$ depends on the user type. If user $k$ is an industrial user then $|\overline s_{\max}(k)| = 1$MVA, otherwise $|\overline s_{\max}(k)| = 5$KVA. More precisely, given a random $r\in |\overline s_{\max}(k)|$, and $f_k({\re(s_k)}) = { r - \tfrac{\re(s_k)}{\re(\overline s_k)} r.   }$

\end{enumerate}	
\end{itemize}
The case studies will be represented by the acronyms. For example, the case study named ``CM'' stands for the one with mixed users and correlated cost-demand setting. 

In order to evaluate the performance of our algorithm PTAS, Gurobi  numerical solver is used as a benchmark to obtain numerical solutions for {\sc OPF}. Note that there is no guarantee that Gurobi will return an optimal solution, nor it will terminate in a reasonable time. Hence, we need to set a timeout (as 200 sec in the evaluation). Whenever Gurobi exceeds the timeout, the current best solution will be considered.
\ifsupplementary
More information of the simulation settings, including parameters of RBTS and IEEE networks, {\color{blue} can be found in the appendix.}
\else
\fi


\begin{figure}[!h]
	\begin{center} 
		\includegraphics[scale=0.5]{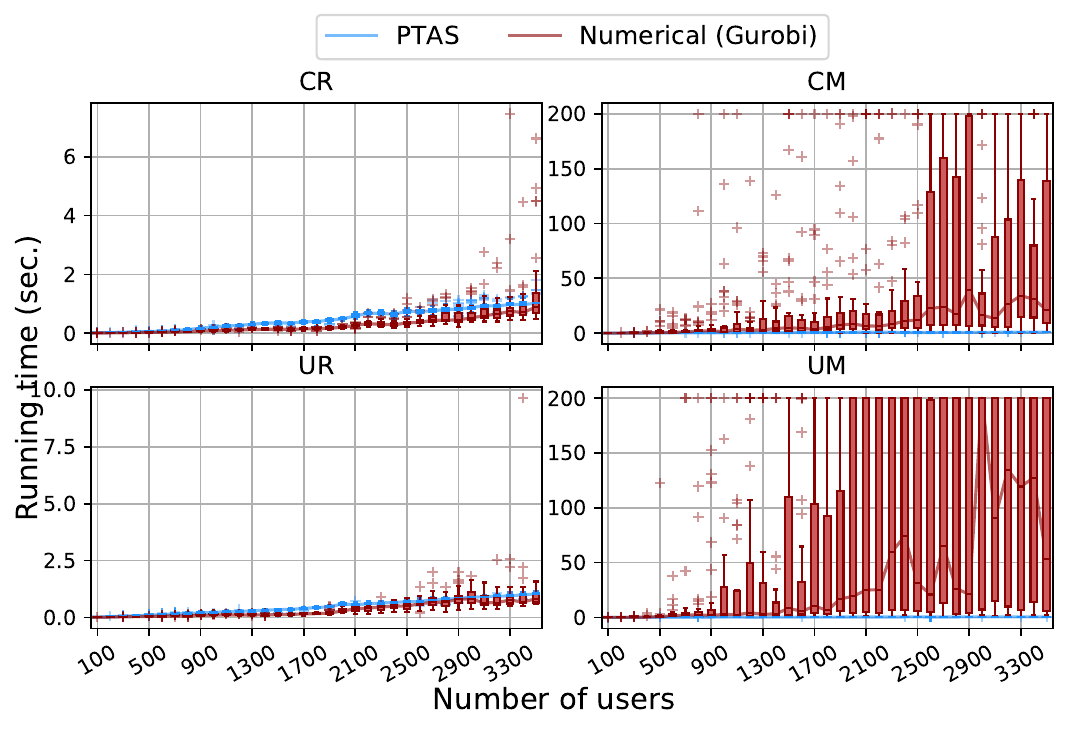}
	\end{center} \vspace*{-15pt}
	\caption{The median of running times of PTAS  (without partial guessing) and Gurobi numerical solver for different case studies in IEEE 123-node network.}\vspace*{-10pt}
	\label{fig:time}
\end{figure}

\begin{figure}[!h]
	\begin{center} 
		\includegraphics[scale=0.48]{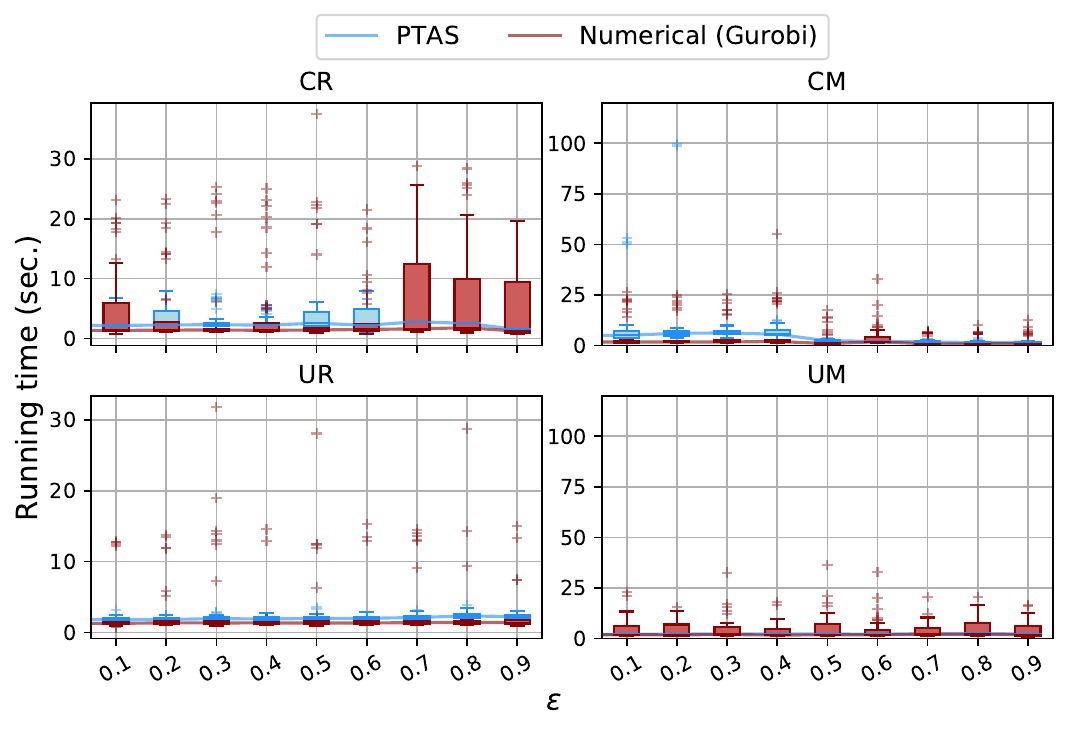}
	\end{center} \vspace*{-15pt}
	\caption{ The median of running times of PTAS and Gurobi numerical solver for different values of $\epsilon$ and different case studies in IEEE 123-node network.}\vspace*{-10pt}
	\label{fig:time_e}
\end{figure}

\vspace*{-10pt}
\subsection{Evaluation Results}
\vspace*{-5pt}
\subsubsection{Optimality}
Fig~\ref{fig:obj13} (resp., \ref{fig:obj123}) presents the objective values obtained by {\sc PTAS}, Gurobi numerical solver, and the lower bounds to the true optimal values by fractional solutions with relaxed discrete demands (i.e., setting all $x_k \in [0,1]$)  respectively using the RBTS 13-node network (resp., IEEE 123-node network) for up to $3500$ users. Each run was evaluated with over $40$ random instances. PTAS will terminate, when its objective value is close to the lower bound.
The objective values of PTAS  are often close to the true optimal values. This is because the number of fractional components in the relaxed problem {\sc P1} is often small. Fig.~\ref{fig:frac} shows the ratio of fractional components over $4m$ is close to 10\%, which stays small when the number of users increases. 

The empirical approximation ratios of PTAS for the two networks are plotted in Fig.~\ref{fig:ar} against the number of users. We observe that the empirical approximation ratio is close to 1.2 in most cases. 
This occurs when the optimal solution satisfies all demands, whereas PTAS (without partial guessing) rounds some fractional demands to zero, which incurs a high cost.
 There are few instances with a larger empirical approximation ratio, but increasing partial guessing is able to resolve this issue,  which is still within polynomial running time. 
 We observe from Fig.~\ref{fig:ar} that the approximation ratio is always below $1.8$,  even without the additional guessing step. Therefore for $\epsilon \ge 0.8$ we may not need guessing step in practice.  
 {Fig.~\ref{fig:ar_e} shows the empirical approximation ratio of {\sc PTAS} {\em with} patial guessing against differnt values for $\epsilon$. In our implementation of {\sc PTAS}, partial guessing is acheived by random sampling, and the algorithm terminates if it either obtains a solution of at most $1+\epsilon$ of the fractional solution, or reaches a given timeout value of $100$. We observe from the figure that the empirical approximation ratio falls below the theoretical (depicted by the dashed line). }

\subsubsection{Running Time}

The computation time of {\sc PTAS} is compared against that of Gurobi numerical solver. Computation time is significantly important when implemented in a controller in practice, and this will have implications to the overall resilience of power grid.
The running time of PTAS (without partial guessing) is plotted in Fig.~\ref{fig:time} under different case studies for IEEE 123-node network.
Although the current implementation of PTAS is not fully optimized, its running time is still substantially better than that of Gurobi, and is observed to scale linearly as the number of users. On the other hand,  the running time of Gurobi is much higher in many cases, which does not provide any guarantee on the termination of execution, if timeout is not set. Many instances experienced timeouts, especially for the case study UR. The actual running time of Gurobi may substantially increase if the timeout value is further increased. 
The running time of {\sc PTAS} is plotted against different values of $\epsilon$ in Fig.~\ref{fig:time_e}. Even though our implementation of {\sc PTAS} is not optimized (implemented in scripting language Python), it performs substantially better than Gurobi (implemented in native language) in most cases. The running time of {\sc PTAS} slightly improves as we increase $\epsilon$. We note that PTAS can be easily parallelized into independent subprograms. Subsequently, the running time could be improved linearly.

\vspace{-10pt}
\section{Conclusion}
This paper presents a polynomial-time approximation algorithm to solve the convex relaxation of OPF with combinatoric constraints, which has a provably small parameterized approximation ratio (also known as PTAS algorithm) that can combine with the sufficient conditions on the exactness of convex relaxation to solve OPF in general. This paper also discusses fundamental hardness results of OPF to show that our PTAS is among the best achievable in theory. Further simulations show that our algorithm can produce close-to-optimal solutions in practice. The topic of combinatorial optimization of AC power networks is burgeoning. Readers can refer to the monograph \cite{chau2018combinatorial} for a collection of recent results. 

\bibliographystyle{ieeetr}
\bibliography{reference}

\end{document}